\documentclass[10pt, a4paper]{amsart}

\usepackage{amsmath}
\usepackage{latexsym}
\usepackage[all]{xy}
\usepackage{color}

\theoremstyle{plain}
\newtheorem{theorem}{Theorem}[section]
\newtheorem{proposition}[theorem]{Proposition}

\newtheorem{lemma}[theorem]{Lemma}
\newtheorem{corollary}[theorem]{Corollary}
\theoremstyle{definition}
\newtheorem{definition}[theorem]{Definition}
\newtheorem{example}[theorem]{Example}

\newcommand{\Diff}[0]{{\textup{Diff}}}

\newcommand{\pd}[2]{\frac{\partial #1}{\partial #2}}

\begin{document}
\begin{title}
{{On $k$-jet field approximations to geodesic deviation equations}}
\end{title}
\maketitle
\author{

\begin{center}

Ricardo Gallego Torrom\'e\footnote{email: rigato39@gmail.com}\\
Departamento de Matem\'atica,\\
Universidade Federal de S\~ao Carlos, Brazil\\
$\&$\\
Jonathan Gratus\footnote{email: j.gratus@lancaster.ac.uk}\\
Physics Department, Lancaster University, Lancaster LA1 4YB\\
and the Cockcroft Institute, UK
\end{center}}

\newtheorem*{theorem-A}{Theorem A}
\newtheorem*{corollary-B}{Corollary B}

\maketitle
\begin{abstract}
Let $M$ be a smooth manifold and $\mathcal{S}$ a semi-spray defined on a sub-bundle $\mathcal{C}$ of the tangent bundle $TM$. In this work it is proved that the only non-trivial $k$-jet approximation to the exact geodesic deviation equation of $\mathcal{S}$, linear on the deviation functions and invariant under an specific class of local coordinate transformations is the Jacobi equation. However, if the linearity property on the dependence in the deviation functions  is not imposed, then there are differential equations whose solutions admit $k$-jet approximations and are invariant under arbitrary coordinate transformations. As an example of higher order geodesic deviation equations we study the first and second order geodesic deviation equations for a Finsler spray.
\end{abstract}

\begin{section}{Introduction}

Given a connection on the manifold $M$ and two neighboring geodesics, it is a fundamental problem in differential geometry and mathematical physics to determine the relative displacement between the pair of geodesics as a function of a common time parameter.  The standard solution to this problem for Lorentzian manifolds is provided by the {\it Jacobi fields}, the solution of the {\it Jacobi equation} along the {\it central Lorentzian geodesic} \cite{Levi-Civita}.

 Despite the geodesic equation is not linear, the Jacobi equation is a linear differential equation. Therefore, the deviation between two geodesics can be described by the solutions of the Jacobi equation only under some restrictions, that sometimes do not hold in interesting physical situations, for instance, when large curvature effects are of relevance. These limitations motivate the search for consistent generalizations of the Jacobi equation.

Currently there are several theories aiming to generalize the Jacobi equation being extensively used in applications of general relativity and in astronomy:
\begin{itemize}
\item The theory developed by D. E. Hodgkinson \cite{Hodgkinson}, based on the {\it linear rapid deviation hypothesis}. In this paper we refer to the fundamental  equation of such theory as the {\it generalized Jacobi equation} (see equation \eqref{eq:genJacobiEq} below). Under similar hypothesis, the theory was later considered by I. Ciufolini \cite{Ciufolini}. The non-covariant form of the generalised Jacobi equation (GJE) equation of Ciufolini coincides with the analogous equation in Hodgkinson's theory. Surprisingly, the corresponding covariant versions are manifestly different.

\item The theory developed by B. Mashhoon \cite{Mashhoon1, Mashhoon:1977}. This theory has been applied to different situations in Astrophysics  \cite{ChiconeMashhoon:2002, ChiconeMashhoon:2005, ChiconeMashhoon:2006}. Although the aim is the same as in Hodgkinson's theory, namely, to have a theory applicable in situations of large geodesic deviations, the meaning of the generalized Jacobi equation in Mashhoon's theory is not equivalent to Hodgkinson's theory. Mashhoon's construction makes extensive use of Fermi coordinates, and the dynamical variables are defined using Fermi coordinates as starting point.

\item The theory  developed by Ba{\.z}a{\'n}ski and others, based on equations for higher order jet fields (see for instance \cite{Bazanski1977, Bazanskikostyukovich1987, Bazanskikostyukovich1987b, KernerColistetevanHolten:2001, ColisteteLeygnacKerner2002, KoekoekHolten2011}). This formalism does not make use of the {\it linear rapid deviation hypothesis} as in Hodgkinson's theory or of specific coordinate constructions as in Mashhoon's theory.

\end{itemize} The problem that we shall consider is formulated at a local level and in a concrete geometric framework.  In the first part of this work we consider the issue of the general covariance of the Hodgkinson's theory in the general setting of geodesics of arbitrary sprays. A main difficulty of this theory consists on the identification of the geometric character of the solutions. For instance, it is known that the solutions of the generalized Jacobi equation of Hodgkinson-Ciufolini when written in the local coordinate form cannot be tensorial \cite{DahlTorrome}. Partially motivated by this result, the covariant  character of the equation when the solutions of the generalized Jacobi equation as $k$-jet fields is investigated here. In fact, generalizing the result from \cite{DahlTorrome}, we have found that for an arbitrary spray, if the solutions of the generalized geodesic deviation equation are approximated by $k$-jet fields, then Hodgkinson's theory is not general covariant. Indeed, this conclusion holds also in the general case of Hodgkinson's theory associated to a semi-spray.

  The proof of the lack of general covariance of Hodgkinson's theory is obtained in the framework of a notion of general covariance  specially suited to our problem. Two key elements of this definition is that it is formulated entirely in terms of local coordinate transformations. Armed with such definition, we show that Hodgkinson's theory is not general covariant, first recognizing the assumptions under which Hodgkinson's equation is necessarily obtained and then showing that under some specific coordinate transformations, such assumptions break down.

In the second half of this paper we show how Ba{\.z}a{\'n}ski's formalism \cite{Bazanski1977, Bazanskikostyukovich1987, Bazanskikostyukovich1987b} can {be} generalized to geodesic differential equations associated with  arbitrary sprays. We apply the {theory} to the case when the spray corresponds to the geodesic equations of a pseudo-Finsler structure, obtaining the second order geodesic deviation equation for pseudo-Finsler structures. We apply this general theory in two situations, for a family of geodesics in Berwald-like structures and for  pseudo-Finsler structures which are not pseudo-Riemannian.
\end{section}

\begin{section}{Preliminary considerations}
\subsection*{Notation}
Let $M$ be an $n$-dimensional smooth manifold and $\mathcal{C}\subset TM$ a sub-bundle of codimension zero such that each fiber $\mathcal{C}_{\bf x}$ is a sub-manifold of $T_{\bf x}M$, with $\pi:\mathcal {C}\to M$ being the canonical projection of the tangent bundle $TM$ restricted to $\mathcal{C}$. A coordinate chart of $M$ is a pair
\begin{align*}
(U,\{\varphi^\mu:U\to \mathbb{R},\,{\bf x}\mapsto x^\mu,\,\mu=1,...,n\}).
 \end{align*}
 Coordinate indices are indicated by Greek characters, and run from $1$ to $n=dim(M)$. The induced local coordinate charts on $\mathcal{C}$ are $(\,\pi^{-1}(U),(\varphi^\mu,\dot{\varphi}^\nu))$. Repeated up and down Greek indices indicates the sum over $1$ to $n$ on each of the values of the index, if anything else is not stated. Since the nature of the problem is to determine the geometric character of the generalized Jacobi equation, we shall work using local coordinate expressions. In order to distinguish between points of $M$ and their local coordinate representations, we shall denote points on manifolds by bold variables; tangent vectors to curves are indicated by a dot over variables in bold, for instance $\dot{\bf x}$; maps with image on a manifold are also denoted by bold characters, like ${\bf x}:I\to M$ for a generic geodesic, etc... Coordinates and components are denoted with un-bold  characters. For instance, the components of a tangent vector are simply $\dot{x}^\mu$. Similarly,  $x^\mu=\varphi^\mu({\bf x})$ denotes local coordinates of {\bf x} and  $\Gamma^\mu_{\nu\rho}({\bf x},\dot{\bf x})$ denotes connection coefficients, etc...
\subsection*{Geodesic equation of a spray}
\begin{definition}
Let $\mathcal{C}\subset \,TM$ be a sub-bundle of $TM$. A second order differential equation (or semi-spray) is a smooth vector field $\mathcal{S}\in\Gamma\,{T}\mathcal{C}$ such that
\begin{align}
\pi_* (\mathcal{S}|_{\bf u})={\bf u},\quad \forall \, {\bf u}\in \mathcal{C}.
\label{semispray}
\end{align}
\end{definition}
Given a local coordinate chart $(\pi^{-1}(U),(\varphi^\mu,\dot{\varphi}^\mu))$ of $\mathcal{C}$, a semi-spray can be expressed in the form
\begin{align}
\mathcal{S}({\bf x},\dot{\bf x})=\,\dot{x}^\mu\frac{\partial}{\partial x^\mu}\Big|_{({\bf x},\dot{\bf x})}-2\,\mathcal{S}^\mu({\bf x},\dot{\bf x})\frac{\partial}{\partial \dot{x}^\mu}\Big|_{({\bf x},\dot{\bf x})},\quad \mu=1,...,n,
\label{localexpressionforS}
\end{align}
where $(x,\dot{x})$ are the coordinates of the point $({\bf x},\dot{\bf x})$ in the local coordinate chart $(\pi^{-1}(U),(\varphi^\mu,\dot{\varphi}^\nu))$.
The geodesics of a semi-spray $\mathcal{S}$ are the projection on $M$ of the integral curves of $\mathcal{S}$. Thus for a semi-spray $\mathcal{S}$, the geodesics are locally the solutions of the system of ordinary differential equations
\begin{align}
\frac{d {x}^\mu}{dt}=\dot{x}^\mu,\quad \frac{d \dot{x}^\mu}{dt}=\,-2\,\mathcal{S}^\mu({\bf x},\dot{\bf x}),\quad \mu=1,...n.
\label{geodesicequationofS}
\end{align}
If the components $\mathcal{S}^\mu({\bf x},\dot{\bf x})$ are positive homogeneous functions of degree $2$ on the coordinates $\dot{x}$ in the sense that $\mathcal{S}^\mu({\bf x},\lambda\,\dot{\bf x})=\,\lambda^2\,\mathcal{S}^\mu({\bf x},\dot{\bf x})$ for any $\lambda\,>0$,  the semi-spray $\mathcal{S}$ is said to be a {\it spray}. In this case, the sub-bundle $\mathcal{C}$ can be considered to be an open cone bundle over $M$ and  Euler's theorem of homogeneous functions implies the relation
\begin{align}
2\,\mathcal{S}^\mu({\bf x},\dot{\bf x})=\,\dot{x}^\sigma\frac{\partial \mathcal{S}^\mu({\bf x},\dot{\bf x})}{\partial \dot{x}^\sigma},\quad \mu=1,...,n.
\label{Sgammarelation}
\end{align}
The functions
\begin{align}
\Gamma^\mu_{\nu}({\bf x},\dot{\bf x}):=\,\frac{\partial \mathcal{S}^\mu({\bf x},\dot{\bf x})}{\partial \dot{x}^\nu},\quad \mu,\nu=1,...,n
\label{Sgammarelation2}
\end{align}
are {\it the non-linear connection coefficients} of the spray $\mathcal{S}$.
Sprays  have the relevant property  that  their integral curves are re-parametrization invariant, a condition which is lost for a general semi-spray. Henceforth we restrict the considerations made in this paper to connections and geodesics associated with sprays only.

 By the relations \eqref{Sgammarelation}-\eqref{Sgammarelation2}, the geodesic equation \eqref{geodesicequationofS} is equivalent to the following system of second order ordinary differential equations,
\begin{align}
\ddot{x}^{\mu}\,+\,\Gamma^{\mu}_{\sigma}({\bf x},\dot{\bf x})\,\dot x^{\sigma}=0,\quad \mu =1,...,n.
\label{nonlineargeodesicequation}
\end{align}
 The $2$-homogeneity property of the spray  $\mathcal{S}$ implies that the non-linear connection coefficients $\Gamma^\mu_\nu({\bf x},\dot{\bf x})$ are homogeneous of degree one in ${\bf \dot{x}}$. Hence by Euler's theorem of homogeneous functions,
\begin{align*}
\Gamma^\mu_\nu({\bf x},\dot{\bf x})=\,\dot{x}^\rho\frac{\partial \Gamma^\mu_{\nu}({\bf x},\dot{\bf x})}{\partial \dot{x}^\rho},\quad ({\bf  x},\dot{\bf x})\in\, T_{\bf x}M,\quad \mu,\nu=1,...,n
\end{align*}
holds good.
Moreover, there is associated to $\mathcal{S}$ a canonical, linear, torsion-free connection $\nabla$ of the tangent bundle $\pi_{\mathcal{C}}: T\mathcal{C}\to \mathcal{C}$ (see for example   \cite{MironHrimiucShimadaSabau:2002}, {\it Chapter} 1).
The connection
\begin{align}
\Gamma^\mu_{\nu\rho}({\bf x},\dot{\bf x}):=\,\frac{\partial \Gamma^\mu_{\nu}({\bf x},\dot{\bf x})}{\partial \dot{x}^\rho},\quad \mu,\nu,\rho=1,...,n,\quad \,({\bf x},\dot{\bf x})\in\, T_{\bf x}M
\label{definitionofgamma}
\end{align}
define an affine connection if and only if $\Gamma^\mu_{\nu\rho}({\bf x},\dot{\bf x})=\Gamma^\mu_{\nu\rho}({\bf x})$ for each
$\,\mu,\nu,\rho=1,...,n$ and at each point ${\bf x}\in M$.
This is the case when $\mathcal{S}^\mu({\bf x},\dot{\bf x})$ is quadratic in $\dot{x}$-coordinates and
  hence $\Gamma^\mu_{\nu}({\bf x},\dot{\bf x})$ is linear in $\dot{x}$-coordinates. However,
for a general spray, the connection $\nabla$ determined locally by the set of connection coefficients $\{\Gamma^\mu_{\nu\rho}({\bf x},\dot{\bf x}),\mu,\nu,\rho=1,...,n\}$  does not live on $M$ and it is not an affine connection.  Relevant examples of non-affine sprays are found in the geometric theory of Finsler spaces,  where the coefficients $\Gamma^\mu_{\nu\rho}({\bf x},\dot{\bf x})$ of the associated connections depend upon the point ${\bf x}\in\,M$ and the direction $\dot{\bf x}\in\,T_{\bf x} M$ in the tangent space.

The connection $\nabla$ of a spray is {\it symmetric}, that is, the relation
\begin{align}
\Gamma^{\mu}_{\nu\sigma}({\bf x},\dot{\bf x})=\,\Gamma^{\mu}_{\sigma\nu}({\bf x},\dot{\bf x}),\quad\, \mu, \nu,\sigma=1,...,n
\label{symmetrycondition}
\end{align}
 holds good. This relation is covariant and it holds in any coordinate system. Furthermore, because of the homogeneity of the connection coefficients on $\dot{x}$-coordinates,
the geodesic equation \eqref{nonlineargeodesicequation} of a spray $\mathcal{S}$ can be re-casted in the form
\begin{align}
\ddot{x}^{\mu}\,+\,\Gamma^{\mu}_{\nu\sigma}({\bf x},\dot{\bf x})\,\dot x^{\nu}\dot x^{\sigma}=0,\quad\mu=1,...,n.
\label{geodesicequation}
\end{align}
This is the canonical local form of the geodesic equations that we mainly consider in this work.

\subsection*{Exact geodesic deviation equation of a spray}
Let us consider two  geodesics ${\bf x}\colon I\to M$ and $ {\bf X}\colon I\to M, \, I\subset \mathbb{R}$ and assume, in order to simplify the treatment, that  ${\bf x}(I)\subset V$ and ${\bf X}(I)\subset V$. Let
$\{\xi^\mu:I\to \mathbb{R},\,\mu=1,...,n\}$ be the coordinate displacement between the geodesics defined as
\begin{eqnarray}
\label{eq:xiDef}
\xi^\mu:I\to \mathbb{R},\quad s\mapsto \xi^\mu(s) := x^\mu(s) - X^\mu(s), \quad s\in I.
\end{eqnarray}
Since ${\bf x}:I\to M$ and ${\bf X}:I\to M$ are  solutions of the geodesic equation \eqref{geodesicequation}, the expression $\{(\nabla_{\dot{\bf X}}\dot{\bf X})^\mu\,-(\nabla_{\dot{\bf x}}\dot{\bf x})^\mu=0 ,\,\mu=1,...,n\}$ can be written in local coordinates as
\begin{align}
\ddot\xi^{\mu}\,+\,\Gamma^{\mu}_{\nu\sigma}(X+\xi,\dot{X}+\dot{\xi})\,
\Big( \dot X^{\nu}+ \dot \xi^{\nu} \Big)
\Big( \dot X^{\sigma}+ \dot\xi^{\sigma}\Big)
 -\,\Gamma^{\mu }_{\sigma\nu}(X,\dot{X})\,
\dot X^{\sigma} \,\dot X^{\nu}=0.
\label{standardexactdeviationequation}
\end{align}
In this expression the connection coefficients $\Gamma^{\mu}_{\nu\sigma}(X+\xi,\dot{X}+\dot{\xi})$ (resp. $\Gamma^{\mu }_{\sigma\nu}(X)$)  depend on the coordinates $(X+\xi,\dot{X}+\dot{\xi})$ of the point $({\bf x}(s),\dot{\bf x}(s))\in \mathcal{C}$ (resp. $\Gamma^{\mu }_{\sigma\nu}(X,\dot{X})$ depend on the coordinates of $(X,\dot{X})$ of the point $({\bf X}(s),\dot{\bf X}(s))\in \,\mathcal{C})$.

The relation \eqref{standardexactdeviationequation} is referred as the \emph{exact geodesic deviation equation}. It is intrinsically a non-geometric relation, since the possibility of its formulations depends upon a particular class of local coordinate systems, namely, the ones containing on the domain $V\subset M$ the graphs ${\bf X}(I),\,{\bf x}(I)$. However, if a local change of coordinates is such that the new domain contains the geodesics ${\bf x}(I) ,{\bf X}(I)$, then the condition \eqref{standardexactdeviationequation} is consistent. Despite its non-local character and intrinsic non-geometric nature, the relation \eqref{standardexactdeviationequation} leads, under suitable approximations, to geometrically well-defined ordinary differential equations.

Given a fixed geodesic ${\bf X}:I\to M$, let us introduce the function $H:\mathbb{R}^{3n}\to \mathbb{R}^n,$
\begin{align}
(\xi^\mu,\zeta^\mu,\vartheta^\mu)\mapsto \vartheta^{\mu}+\Gamma^{\mu}_{\nu\sigma}(X+\xi,\dot{X}+{\zeta})\,
\Big( \dot X^{\nu}+  \zeta^{\nu} \Big)
\Big( \dot X^{\sigma}+ \zeta^{\sigma}\Big)
 -\,\Gamma^{\mu }_{\sigma\nu}(X,\dot{X})
\dot X^{\sigma} \,\dot X^{\nu}.
\label{equationforD}
\end{align}
Then the exact geodesic deviation equation \eqref{standardexactdeviationequation} reads as the condition
 \begin{align*}
 H(\xi^\mu(s),\dot{\xi}^\mu(s),\ddot{\xi}^\mu(s))=0,
 \end{align*}
  where $H$ is interpreted as a function along the geodesic $X:I\to M$.

\subsection*{A restricted notion of general local covariance}
The set of diffeomorphisms of $\mathbb{R}^n$ is denoted by ${\bf \Diff}(\mathbb{R}^n)$. ${\bf \Diff}(\mathbb{R}^n)$ has the structure of an infinite dimensional Lie group.
Let ${M}$ be a smooth manifold of dimension $n$ and let us consider a local coordinate chart  $(V,\varphi)$ of $M$. Then to each element $\chi\in\,{\bf \Diff}(\mathbb{R}^n)$ there is associated a transformation of local coordinates on $M$ by the map $(V,\varphi)\to (V,\chi\circ\varphi)$, where the domain of $\chi$ is restricted to $\varphi(V)\subset \,\mathbb{R}^n$. In doing this we are considering the restriction $\chi|_{\varphi(V)}:\varphi(V)\to \mathbb{R}^n$ as an element of the pseudo-group $\Gamma(\mathbb{R}^n)$ \cite{KobayashiNomizu1963}.  Indeed, $\Gamma(\mathbb{R}^n)$ is in one-to-one correspondence with the local coordinate transformations of $M$ with the same chart domain $V$.

 A {\it classical physical model} will be represented in a local coordinate system $(V,\varphi)$ of $M$ by a collection of maps
$ \{\psi_A\circ \varphi^{-1}:\varphi(V)\to \mathbb{R}^n,\,A=1,...,k\}.$
Let us consider  an arbitrary element $\chi\in\,\Gamma(\mathbb{R}^n)$ $x\mapsto \tilde{x}=\chi(x).$
\begin{definition}
 The coordinate representation
 $\{\tilde{\psi}_A\circ \tilde{\varphi}^{-1}:\tilde{\varphi}(\tilde{V})\to \mathbb{R}^n,\,A=1,...,k\}$ of a physical variable $\psi$
in the local coordinate system $(\tilde{V},\tilde{\varphi})$ is obtained from the coordinate representation
 $\{\psi_A\circ \varphi^{-1}:\varphi(\tilde{V})\to \mathbb{R}^n,\,A=1,...,k\}$ in the local coordinate system $(V,\varphi)$
 by the action of $\chi\in \,\Gamma(\mathbb{R}^n)$ if three conditions hold:
 \begin{enumerate}

 \item $V\cap\tilde{V}\neq \emptyset,$

 \item $\tilde{\varphi}=\,\chi\circ\varphi$,

 \item There exits a function $\{r_{\chi}:\mathbb{R}^k\to\mathbb{R}^k\}$ such that
 \begin{align}
 \widetilde{\psi}_B\circ \tilde{\varphi}^{-1}=\,r_\chi({\psi}_1\circ \varphi^{-1},...,{\psi}_k\circ \varphi^{-1}),\quad B=1,...,k
 \label{actionofgamma}
 \end{align}
 holds good in $\varphi(V)\cap\,\tilde{\varphi}(\tilde{V}),$ such that $r_\chi$ is a representation of $\Gamma(\mathbb{R}^n)$: for each pair $\chi_1, \chi_2\in \,\Gamma (\mathbb{R}^n)$ it holds that $r_{\chi_1 \cdot \chi_2}=\,r_{\chi_1}\circ r_{\chi_2}$, whenever the composition $r_{\chi_1 \cdot \chi_2}$ is defined.
 \end{enumerate}
 \label{change of local representations}
\end{definition}
In order that the above {\it definition} can be applied, an specific form for $r_\chi$ must be assumed. These laws depend on the specific theory that we could be interested. For instance, the functions $\psi^A\mapsto r_\chi(\psi^A)$ could define a tensorial or a density-like or spinorial quantities. However, the representation \eqref{actionofgamma} is not restricted to these ones. Indeed, another relevant example is the case when the coordinate representation $\{\psi_A,\,A=1,...,k\}$ transforms under coordinate transformation as a $k$-jet.
\begin{definition}
A system of $m$-equations  depending on the local coordinates
$\{x^{\mu}_1,...,x^{\mu}_p\}$
 of points $\{{\bf x}_1,...,{\bf x}_p\}\subset V$ and  components of local fields $\{\psi_{A j}\}^{k,p}_{A=1,j=1}$ written in local coordinates $(V,\varphi)$ as
\begin{align*}
G_i\big(x_1,...,x_p,\psi_{A1}( { x}_1),...,\psi_{Ap}({x}_p)\big)=0,\quad i=1,...,m
\end{align*}
 is said to be {\it general  local covariant}
 iff for each $\chi\,\in\, \Gamma(\mathbb{R}^n)$, the equations
 \begin{align*}
G_i\big(\chi({x}_1),...,\chi({x}_p),\tilde{\psi}_{{A1}}(\chi({x}_1)),...,\tilde{\psi}_{Ap}(\chi({x}_{p})) \big)=0,\quad i=1,...,m
 \end{align*}
 hold in the local coordinate chart $(V,\tilde{\varphi}=\chi\circ \varphi)$, where $\tilde{\psi}_{A_i}$ are defined by a rule \eqref{actionofgamma} according to a given set of representations $r_{A_i}$ of $\Gamma(\mathbb{R}^n)$.
 \label{definicionrestrictedlocalcovariance}
\end{definition}
 This restricted notion of {\it general local covariance} can be extended  straightforwardly to inequalities of the type ${G_i(x,\psi_A)<0,\,i=1,...,m}$, where ${G_i}$ is an expression determined by local objects. Although it is not the most general notion that one could found, our notion of general local covariance embraces the usual notions of covariance, can be applied to non-local equations, that is, equations depending on more than one point on the manifold and it will be general enough for our purposes.

 We should observe how restricted is this notion of general covariance given above. A key point of our {\it definition} is that the domain $V\subset M$ is preserved in transformations considered. In practical terms, that means a restriction on the type of local coordinate changes that we shall consider.

\subsection*{Approximations of the exact deviation equation}
The main motivation to consider this restricted version of general covariance is the following. Under the constraint that the domain $V\subset M$ does not change under local coordinate transformations of the type considered,
the expression  \eqref{standardexactdeviationequation} is invariant. This is because it is obtained as the difference between the components of two zero vectors, the acceleration vectors along the geodesic ${\bf x}:I\to M$ and along ${\bf X}:I\to M$. This preservation of the domain is a sufficient condition for preserving the physical meaning of the  geodesic deviation relation \eqref{standardexactdeviationequation} as the local displacement coordinate functions between two geodesics. Otherwise, if the domain $V$ changes or is restricted under the coordinate transformation, it could be that the geodesics lay out of the new domain.
However, the relation \eqref{standardexactdeviationequation}  is non-local in $M$, since the definition of $\{\xi^\mu:I\to M,\,\mu=1,...,n\}$ involves for each $s\in I$ the local coordinates of two points on the manifold $M$.

In order to understand better the relative behavior of geodesics, one can introduce several approximations:
\begin{itemize}
\item The connection coefficients $\Gamma^{\mu}_{\nu\sigma}(X+\xi,\dot{X}+\dot{\xi})$ are approximated by Taylor's series  in terms of the functions $\xi^\mu$ and $\dot{\xi}^\mu$ at the point $\xi^\mu_0=0$ and $\dot{\xi}^\mu_0=0$.

\item $\{\xi^\mu,\dot \xi^\mu\}^n_{\mu=1}$ are {\it infinitesimal}, in the sense that the terms proportional to $\{\xi^\mu\xi^\nu, \,\xi^\mu\dot{\xi}^\nu,\text{etc}...\}$ can be disregarded.
\end{itemize}
 Then the relation \eqref{standardexactdeviationequation} under these  approximations yields to the expression
\begin{eqnarray}
\ddot \xi^{\mu}_1\,+\pd{\Gamma^{\mu}_{\nu\sigma}}{x^\rho} (X,\dot{X})\,\xi^{\rho}_1
\dot X^{\nu}\,\dot X^{\sigma}
+2\,\Gamma^{\mu }_{\nu\sigma}(X,\dot{X})\, \dot\xi^{\nu}_1\dot X^{\sigma}=0,
\label{Jacobiequation}
\end{eqnarray}
that can be thought as an ordinary differential equation for the variables $\{\xi^i,\,i=1,...,n\}$. Equation \eqref{Jacobiequation} is linear and the set of solutions defines the finite rank vector space of {\it Jacobi fields} along ${\bf X}:I\to M$. Indeed, equation \eqref{Jacobiequation} is a non-explicit covariant way of writing the {\it Jacobi equation} \cite{Levi-Civita} of the connection $\nabla$,
\begin{eqnarray}
    \nabla_{\dot{\bf X}}\nabla_{\dot{\bf X}}{\bf J} + \mathcal{R}({\bf J},\dot{\bf X})\dot{\bf X} = 0,
 \label{covariantformjacobiequation}
\end{eqnarray}
where ${\bf J}(s)=\,\xi^\mu_1(s)\frac{\partial }{\partial x^\mu}\big|_{{\bf X}(s)}\in T_{{\bf X}(s)}M$ and $\mathcal{R}({\bf J},\dot{\bf X})$ is the Riemann type curvature endomorphism  of the connection $\nabla$ determined by ${\bf J},\dot {\bf X}\in T_{{\bf X}(s)}M$. Thus equation \eqref{Jacobiequation} is compatible with the action of the pseudo-group $\Gamma(\mathbb{R}^n)$ and is local in the sense that it is formulated in terms of objects defined along the geodesic ${\bf X}:I\to M$.

If only the deviation functions $\{\xi^\mu,\,\mu=1,...,n\}$ are infinitesimal in the sense that the only quadratic monomials which are negligible are $\{\xi^\mu\xi^\nu,\,\mu,\nu=1,...,n\}$, then equation \eqref{standardexactdeviationequation}  yields to the
\emph{generalized Jacobi equation} \cite{Hodgkinson},
\begin{eqnarray}
\label{eq:genJacobiEq}
\ddot \kappa^{\mu}
+ \Gamma^{\mu}_{\rho\nu} \left( 2 \dot \kappa^{\rho} \dot X^{\nu} + \dot \kappa^{\rho} \dot \kappa^{\nu}\right)
+ \pd{\Gamma^{\mu}_{\rho \nu}}{x^\sigma} \kappa^\sigma \left( \dot X^{\rho}
 + \dot \kappa^{\rho}\right) \left( \dot X^{\nu} + \dot \kappa^{\nu}\right) = 0.
\end{eqnarray}
In this equation, the unknown variables $\{\kappa^\mu\}^n_{\mu=1}$ have an obscure geometric interpretation and motivates the problem of understanding the geometric character of equation \eqref{eq:genJacobiEq}, in particular the question of its {\it general covariance} character is not clear.  The general covariance problem is equivalent to the compatibility of the equation \eqref{eq:genJacobiEq} with the action of the pseudo-group $\Gamma(\mathbb{R}^n)$ on the physical relevant fields.
In practice we shall investigate a weaker condition of general local covariance, based on our {\it definition}  \ref{definicionrestrictedlocalcovariance}. We will find a negative answer even to this restricted covariance issue.

Analogous considerations hold for the general case of semi-spray geodesics. Different Taylor's expansions can be applied to the components $\mathcal{S}^\mu$ of the semi-spray, that leads to the analogous of the equations \eqref{Jacobiequation} and \eqref{eq:genJacobiEq}. The theory and results developed below apply to the general case of semi-spray geodesics deviation equations.
\subsection*{Jet fields approximations}
 Motivated by the difficulties in determining the geometric properties of equations \eqref{eq:genJacobiEq}, it is reasonable to assume first an specific geometric character for its solutions. A very natural possibility is to assume that $\kappa^\mu$ are components of a vector field. However, it was proved that in this case $\{\kappa^{\mu}\}$ cannot be tensorial \cite{DahlTorrome}. Other natural possible interpretation  for the solutions of equation \eqref{eq:genJacobiEq} is that they are smooth maps $\Psi_k:I\to J^k_0(I,M)$  such that the diagram
\begin{align}
\xymatrix{ &
{ J^k_0(I,M)} \ar[d]^{\pi_{k}}\\
{ I} \ar[ur]^{\Psi_k}  \ar[r]^{{\bf X}} & { M}}
\label{commutativediagram}
\end{align}
commutes,  where $J^k_0(I,M)\to M$  is the $k$-jet space and $\pi_k:J^k_0(I,M)\to M$ is the canonical projection on $M$. We call the maps $\Psi_k:I\to J^k_0(I,M)$ $k$-{\it jet fields}. Some relevant fields defined along a given geodesic ${\bf X}:I\to M$ are of this type. For instance, for  $k=1$,  $J^1_0(I,M)$ could correspond to the restriction of  the tangent bundle $\pi_1:TM\to M$ along ${\bf X}:I\to M$  and the sections  $\{\xi^{\mu}_1:I\to J^1_0(I,M)\}$ along ${\bf X}:I\to M$ the Jacobi vector fields ${\bf J}(s)=\,\xi^\mu_1(s)\frac{\partial }{\partial x^\mu}|_{{\bf X}(s)}$, solutions of the Jacobi equation \eqref{covariantformjacobiequation}.

It is convenient to work with $1$-parameter geodesic variation
 associated with the pair of geodesics ${\bf X},{\bf x}:I\to M$. Given a generic spray, a $1$-parameter geodesic variation is a map
\begin{align*}
{\bf \Lambda}:(-\epsilon_0,\epsilon_0)\times I\to M,\quad \epsilon_0>0
\end{align*}
such that the following three requirements are fulfilled:
\begin{itemize}
\item For each $\epsilon\in\,(-\epsilon_0,\epsilon_0)$ the curve ${\bf \Lambda}(\epsilon,\cdot):I\to M$ is a geodesic,

\item ${\bf \Lambda}(0,s)={\bf X}(s)$,

\item There exists a value $\bar{\epsilon}\in\, (-\epsilon_0,\epsilon_0)$ such that ${\bf \Lambda}(\bar{\epsilon},s)={\bf x}(s)$.
\end{itemize}
Then the functions $\Lambda^\mu(\epsilon,s)$ can be expanded in the first variable $\epsilon\in \,(-\epsilon_0,\epsilon_0)$ by Taylor's theorem, allowing to approximate the fields $\{\kappa^\mu: I\to M,\,\mu=1,...,n\}$ in terms of $k$-jet sections and perform an analysis order by order in $\epsilon$ of the general local covariance of the equation \eqref{eq:genJacobiEq}.

 By embedding the geodesics ${\bf X}:I\to M$ and ${\bf x}:I\to M$ in the {ribbon} ${\bf \Lambda}((-\epsilon_0,\epsilon_0)\times I)\subset \,M$ we can approximate the functions $\{\xi^\mu:I\to \mathbb{R},\mu=1,...,n\}$, solutions of the exact deviation equation \eqref{standardexactdeviationequation}, by the Taylor's expansions
\begin{align*}
\xi^\mu(s):=\Omega_k^\mu(\bar{\epsilon},s)+\mathcal{O}(\bar{\epsilon}^{k+1})\,=\,\sum^k_{i=0}\,\bar{\epsilon}^i\,\frac{1}{i!}\Lambda^\mu_i(s)
-\Lambda^\mu(0,s)+\,
\mathcal{O}(\hat{\bar{\epsilon}}^{k+1}),
\end{align*}
for  fixed $\bar{\epsilon},\hat{\bar{\epsilon}}\in\,(-\epsilon_0,\epsilon_0)$. Indeed, it is useful for a variational interpretation of the deviation functions $\xi^\mu$ to consider the following {\it ribbon coordinate functions},
\begin{align}
\Xi^\mu:(-\epsilon_0,\epsilon_0)\times I\to \mathbb{R},\quad (\epsilon,s)\mapsto \Lambda^\mu(\epsilon,s)-\Lambda^\mu(0,s).
\label{ribboncoordinatefunctions}
\end{align}
The fields $ {\Omega} ^\mu_k(\epsilon,s)$ and $\varrho^\mu_{k+1}(s)$ are given by the relations
\begin{align}
{\Omega} ^\mu_k(\epsilon,s)=\Lambda^\mu(\epsilon,s)-\Lambda^\mu(0,s)-\,\frac{ \epsilon^{k+1} }{(k+1)!}\,\varrho^\mu_{k+1}(s):=\,\sum^k_{i=1}\, \frac{\epsilon^i}{i!}\,\xi^\mu_i(s),\quad\mu=1,...,n,
\label{omegakfield}
\end{align}
where the {\it remainder fields} are
  \begin{align}
  {\varrho}^{\mu}_{k+1}(s):=\,\frac{\partial^{k+1} \Lambda^\mu(\epsilon,s)}{\partial \epsilon ^{k+1}}\Big|_{\epsilon=\hat{\epsilon}(s)},\quad\mu=1,...,n.
   \label{remainderfields}
   \end{align}
 The corresponding first time derivatives are
\begin{align*}
\dot{\Omega} ^\mu_k(\epsilon,s):=\,\dot{\Lambda}^\mu(\epsilon,s)-\dot{\Lambda}^\mu(0,s)-\,\frac{ \epsilon^{k+1} }{(k+1)!}\,\dot{\varrho}^\mu_{k+1}(s)=\,\sum^k_{i=1}\, \frac{\epsilon^i}{i!}\,\dot{\xi}^\mu_i(s),\quad\mu=1,...,n.
\end{align*}
The fields
\begin{align*}
 \left\{\xi^{\mu}_j(s):=\,\frac{\partial^j \Lambda^\mu(\epsilon,s)}{\partial \epsilon ^j}\Big|_{\epsilon=0},\mu=1,...,n,\,j=1,...,k\right\}
  \end{align*}
  as defined by the relation \eqref{omegakfield} live along the central geodesic ${\bf \Lambda}(0,s)={\bf X}(s)$. ${\Omega} ^\mu_k$ are $k$-jet fields along ${\bf X}:I\to M$. On the other hand, the remainder fields \eqref{remainderfields}
 do not live along the geodesic ${\bf X}:I\to M$.

The errors in the Taylor's approximations
\begin{align*}
\Xi^\mu(\epsilon,s)\to\Omega^\mu_k(\epsilon,s),\quad \dot{\Xi}^\mu(\epsilon,s)\to\dot{\Omega}^\mu_k(\epsilon,s),\quad \mu=1,...,n
\end{align*}
 are given by the remainder terms of the Taylor's expansions, which are of order $\epsilon^{k+1}$. Therefore, the error in the approximation $\xi^\mu(s)$ by $\Omega^\mu_k(\bar{\epsilon},s)$ is of order $\bar{\hat{\epsilon}}^{k+1}$, since the values $\epsilon$ and $\hat{\bar{\epsilon}}$ are both bounded by $\epsilon_0$, which is considered very small compared to $1$.

 There is also a ribbon version of the exact deviation equation \eqref{standardexactdeviationequation},
\begin{align}
\begin{split}
&\ddot\Xi^{\mu}(\epsilon,s)+\,\Gamma^{\mu}_{\nu\sigma}(X+\Xi(\epsilon,s),\dot{X}+\dot{\Xi}(\epsilon,s))\,
\Big( \dot X^{\nu}+ \dot \Xi^{\nu}(\epsilon,s) \Big)
\Big( \dot X^{\sigma}+ \dot\Xi^{\sigma}(\epsilon,s)\Big)\\
& -\Gamma^{\mu }_{\sigma\nu}(X,\dot{X})\,
\dot X^{\sigma} \,\dot X^{\nu}=0.
\end{split}
\label{ribbonstandardexactdeviationequation}
\end{align}
Note that in {\it equation} \eqref{ribbonstandardexactdeviationequation} there are no derivatives respect to $\epsilon$.  Therefore, partial derivatives respect the parameter $s$ are  usual ordinary  derivatives and equation \eqref{ribbonstandardexactdeviationequation} is indeed an usual ordinary differential equation depending on the external parameter $\epsilon$.

The main problem considered in the first part of this paper can now be stated as follows,
\\
{\bf Problem}. To determine if the further approximation
\begin{align}
\{\kappa^\mu(s),\mu=1,...,n\}\to\{\Omega_k^\mu(\bar{\epsilon},s),\,\mu=1,...,n\},
\label{aproximationkappaOmega}
\end{align}
where $\{\kappa^\mu(s),\,\mu=1,...,n\}$ is a solution of the equation \eqref{eq:genJacobiEq} and the parameter $\bar{\epsilon}\in\,(-\epsilon_0,\epsilon_0)$ is fixed by the condition ${\bf \Lambda}(\bar{\epsilon},s)={\bf x}(s)$, is consistent with the action of the pseudo-group $\Gamma(\mathbb{R}^n)$ as defined above.
\bigskip

For arbitrary sprays, we shall see that the answer to this problem is negative. This implies the non-consistency of Hodgkinson's theory with general covariance as formulated following our {\it definition} \ref{definicionrestrictedlocalcovariance}, since the assumptions on which equation \eqref{eq:genJacobiEq} is based are not consistent with general covariance in the framework of $k$-jet fields.

 \subsection*{Embedding of two geodesics in a geodesic variation}
  We will assume that  in the geodesic variation ${\bf \Lambda}:(-\epsilon_0,\epsilon_0)\times I\to M$, the \emph{transversal parameter} $\epsilon\in\,(-\epsilon_0,\epsilon_0)$ is invariant by the action of each element of $\Gamma(\mathbb{R}^n)$.
Let us consider an spray $\mathcal{S}\in \Gamma\,T\mathcal{C}$.
The following result shows that locally, two nearby enough geodesics ${\bf x}:I\to M$ and ${\bf X}:I\to M$  can be described by a $1$-parameter geodesic variation for a shorter time interval $\tilde{I}\subset I$,
\begin{proposition}
Let ${\bf x}:I\to M$ and ${\bf X}:I\to M$ be two geodesics of the semi-spray $\mathcal{S}$ such that ${\bf x}(0)$ and ${\bf X}(0)$ are connected by an extensible, simple transverse curve ${\bf c}$. Then there is a $1$-parameter geodesic variation  ${\bf \Lambda}:(-\epsilon_0,\epsilon_0)\times \tilde{I}\to M$ such that the central geodesic is ${\bf \Lambda}(0,s)=\, X(s)$ and ${\bf \Lambda}(\bar{\epsilon},s)={\bf x}(s)$ for some $\bar{\epsilon}\in\,(-\epsilon_0,\epsilon_0)$.
\label{Proposition}
\end{proposition}
 \begin{proof}
The two initial points ${\bf x}(0)$ and ${\bf X}(0)$ can be joined by the {\it connecting curve} ${\bf c}:[0,\bar{\delta}]\to M$ with ${\bf c}(0)={\bf X}(0)$ and ${\bf c}(\bar{\delta})={\bf x}(0)$.   Indeed, one can extend the curve ${\bf c}$ to include $[0,\bar{\delta}]$ in an open interval $(0-\sigma, \delta_0)$, with $\bar{\delta}<\delta_0$. To obtain the desired geodesic variation we construct an appropriate initial conditions along ${\bf c}$. First, the tangent vector $\dot{{\bf X}}(0)$ is parallel transported along ${\bf c}$ by the parallel transport of the connection $\nabla$, defining a vector field $\hat{\bf X}'(\delta)$ along ${\bf c}$. A similar operation can be done for $\dot{{\bf x}}(0)$ but along the inverted curve $-{\bf c}(s):={\bf c}(\bar{\delta}-s)$ joining the points  ${\bf x}(0)$ and ${\bf X}(0)$, determining a vector field ${\hat{\bf x}}'({\bf \delta})$ along $-{\bf c}(s)$. Then let us consider the linear combination of vector fields along ${\bf c}$,
 \begin{align*}
  {\bf {Z}}(\delta,0)=\,\frac{1}{\bar{\delta}}\,\Big(\delta\,\hat{{\bf x}}'(\delta)+\,(\bar{\delta}-\delta)\hat{{\bf X}}'(\delta)\Big)\in T_{{\bf c}(\delta)}M,\quad \bar{\delta}\neq 0.
  \end{align*}
  By Picard-Lindel\"{o}f's theorem, each initial value $({\bf c}(\delta),{\bf Z}(\delta,0))$ determines an unique geodesic ${\bf \Lambda}(\delta,\cdot):[0,s_{max}(\delta)]\subset I\to M$ for some maximal $s_{max}(\delta)\in I$ positive. By a standard argument using the compactness of the domain $[0,\bar{\delta}]\times I$ it follows that
  \begin{align*}
  I\ni\,\hat{s}_{max}:=\,min\{s_{\max}(\delta),\delta \in\,[0,\bar{\delta}]\}>0.
  \end{align*}
 By continuity, the same is true for $\sigma$ small enough and $s_{max}>0$ defined as
   \begin{align*}
  {s}_{max}:=\,min\{s_{\max}(\delta),\delta \in\,(0-\sigma,\bar{\delta}-\sigma)\}.
  \end{align*}
 We have constructed a geodesic variation ${\bf \Lambda}:(0-\sigma,\bar{\delta}+\sigma)\times [0,s_{max}]\to M$ with ${\bf \Lambda}(0,s)={\bf X}(s)$, ${\bf \Lambda}(\bar{\delta},s)={\bf x}(s)$ for $s\in \,[0,s_{max}]$. By a convenient re-parametrization of ${\bf c}$, the parameter in the variation can be redefined in the interval $(-\epsilon_0,\epsilon_0)$ and still be constrained by the conditions ${\bf \Lambda}(0,s)={\bf X}(s)$ and ${\bf \Lambda}(\bar{\epsilon},s)={\bf x}(s)$ in the new re-parametrization of ${\bf \Lambda}$ with the required properties.
   \end{proof}
   \begin{corollary}
   Let ${\bf x},{\bf X}:I\to M$ and ${\bf \Lambda}:(-\epsilon_0,\epsilon_0)\times I\to M$ be as in {\it Proposition} \ref{Proposition}. Then there exists  an $\bar{\epsilon}\in (-\epsilon_0,\epsilon_0)$ such that the relation $\xi^\mu(s)=\Lambda^\mu(\bar{\epsilon},s)-\Lambda^\mu(0,s)$  holds.
   \end{corollary}
By application of Taylor's theorem up to order $k$ to the function $\Lambda^\mu$, it follows that in local coordinates
\begin{align}
\Lambda^\mu(\epsilon,s)=\,\sum^{k}_{j=0}\,\frac{1}{k!}\,\epsilon^j\,\frac{\partial^j \Lambda^\mu(\epsilon,s)}{\partial \epsilon ^j}\Big|_{\epsilon=0} +\,\frac{ \epsilon^{k+1} }{(k+1)!}\frac{\partial^{k+1} \Lambda^\mu(\epsilon,s)}{\partial \epsilon ^{k+1}}\Big|_{\epsilon=\hat{\epsilon}(s)},
\label{taylorexpansion}
\end{align}
with $ \hat{\epsilon}(s)\in\,(-\epsilon_0,\epsilon_0).$
Taking derivatives respect to the parameter $s$ in \eqref{taylorexpansion} one obtains
\begin{align}
\dot{\Lambda}^\mu(\epsilon,s)=\,\sum^{k}_{j=0}\,\frac{1}{k!}\,\epsilon^j\,\frac{\partial^j \dot{\Lambda}^\mu(\epsilon,s)}{\partial \epsilon ^j}\Big|_{\epsilon=0} +\,\frac{ \epsilon^{k+1} }{(k+1)!}\,\frac{\partial^{k+1} {\dot\Lambda}^\mu(\epsilon,s)}{\partial \epsilon ^{k+1}}\Big|_{\epsilon=\hat{\epsilon}(s)}.
\end{align}

\subsection*{Approximation schemes}
Let us consider a geodesic variation ${\bf \Lambda}:(-\epsilon_0,\epsilon_0)\times I\to M$ associated to the pair of geodesics ${\bf x}:I\to M$ and ${\bf X}:I\to M$ by an embedding as in {\it Proposition} \ref{Proposition}. Furthermore, let us fix an initial local coordinate system on $M$. Then we can apply Taylor's expansions on $\epsilon$ to the smooth maps $\Lambda^\mu:(-\epsilon_0,\epsilon_0)\times I\to \mathbb{R}$ to approximate $\xi^\mu(s)$. An {\it approximation scheme} is a set of negligible monomials at order $k$ of the free algebra $\mathcal{A}(\xi,\dot{\xi},\ddot{\xi})$ generated by  the monomials $\{\xi^\mu,\dot{\xi}^\nu,\ddot{\xi}^\rho,\,\mu,\nu,\rho=1,...,n\}$.

\begin{example} The following four examples are considered in this work:
\begin{itemize}

\item {\it Trivial approximation scheme}, where none
of the monomials
\begin{align*}
\{\xi^\mu{\xi}^\nu,\xi^\mu\dot{\xi}^\nu,\dot\xi^\mu\dot{\xi}^\nu,\,\mu, \nu =1,...,n\}
\end{align*}
 is  negligible. This corresponds to the exact deviation equation \eqref{standardexactdeviationequation}.

\item {\it Linear approximation scheme}, where all the quadratic monomials generated by
\begin{align*}
\{\xi^\mu{\xi}^\nu,\xi^\mu\dot{\xi}^\nu,\dot\xi^\mu\dot{\xi}^\nu,\,\,\mu, \nu =1,...,n\}
\end{align*}
 are negligible. This corresponds to the approximation leading to the Jacobi equation \eqref{Jacobiequation}.

\item {\it Linear rapid deviation scheme}, where
the only negligible monomials in the algebra $\mathcal{A}(\xi,\dot{\xi},\ddot{\xi})$ are generated by
\begin{align*}
\{\xi^\mu{\xi}^\nu,\,\mu,\nu=1,...,n\}.
\end{align*}
This corresponds to the approximation leading to the generalized Jacobi equation \eqref{eq:genJacobiEq} of Hodgkinson's theory \cite{Hodgkinson}.
\item {\it Quadratic approximation}, where from the free algebra $\mathcal{A}(\xi,\dot{\xi},\ddot{\xi})$ only the monomials
\begin{align*}
\{\xi^\mu,\dot{\xi}^\nu,\ddot{\xi}^\rho,\xi^\mu{\xi}^\nu,\xi^\mu\dot{\xi}^\nu,\dot\xi^\mu\dot{\xi}^\nu,\,\mu, \nu =1,...,n\}
\end{align*}
are not negligible. This scheme corresponds to the second order differential equation in Ba{\.z}a{\'n}ski's theory \cite{Bazanski1977}.

\end{itemize}
The approximations schemes can be re-casted in terms of monomials generated by the functions $\{\Xi^\mu\}^n_{\mu=1}$ and their time derivatives and alternatively, in terms of the functions $\{\Omega^\mu_k\}^n_{\mu=1}$ and their times derivatives. Thus it is equivalent to neglect a monomial like   $\xi^\mu\dot{\xi}^\nu$ or to neglect the corresponding {\it ribbon monomial} $\Xi^\mu\dot{\Xi}^\nu$.
\end{example}

The full scheme of approximations behind the $k$-jet solutions  scheme to the generalized Jacobi equation \eqref{eq:genJacobiEq} are the following:
\begin{enumerate}
\item The deviation functions $\{\xi^\mu:I\to \mathbb{R},\,\mu=1,...,n\}$ are substituted by the ribbon coordinate functions $\{\Xi^\mu:(-\epsilon_0,\epsilon_0)\times I\to \mathbb{R},\,\mu=1,...,n\}$. The justification of this fact is based on {\it Proposition} \ref{Proposition}.

\item The ribbon coordinate functions $\{\Xi^{\mu}:(-\epsilon_0,\epsilon_0)\times I\to \mathbb{R}\}$ are approximated by the Taylor's expansions $\{\Omega_k^\mu:(-\epsilon_0,\epsilon_0)\times I\to \mathbb{R}\}$. This approximation can be done with arbitrary accuracy by choosing a high enough order $k$, by Taylor's theorem.
\item One assumes a particular approximation scheme. In particular the linear rapid deviation approximation scheme is equivalent to the approximation $\Omega^\mu_k\to \kappa^\mu$, $\mu=1,...,n$  in the ribbon version \eqref{ribbonstandardexactdeviationequation} of the exact deviation equation\eqref{standardexactdeviationequation}, that  when applied at the parameter $\epsilon=\bar\epsilon$, leads to the generalized Jacobi equation of Hodgkinson.
\end{enumerate}
These three approximations together with the embedding property showed in {\it Proposition} \ref{Proposition} imply that the $k$-jet fields $\{\Omega_k^\mu:(-\epsilon_0,\epsilon_0)\times I\to \mathbb{R}\}$ determine an approximation up to order $\epsilon^{k+1}$ for the solutions  $\{\kappa^\mu:I\to M,\,\mu=1,...,n\}$ of the equation \eqref{eq:genJacobiEq}, for $|\epsilon|<\epsilon_0$.

\section{Method and results}
The consistency criterion to be checked against general covariance according to our {\it definition} \ref{definicionrestrictedlocalcovariance} is that the error introduced in each of the approximations described above must be bounded or be of the same order than the error in the jet field approximation $\{\xi^\mu\to \Omega^\mu_k(\bar{\epsilon}, s),\,\mu=1,...,n\}$. Note that the value $\bar{\epsilon}$ has been fixed. However, in order to have a variational interpretation, we consider the family defined by the variable parameter $\epsilon\in (-\epsilon_0,\epsilon_0)$. If $\epsilon , \bar{\epsilon}\leq \epsilon_0$, the error in the approximation is of the same order for $\epsilon$ and $\bar{\epsilon}$. The physical motivation for this extension is that we do not only require the deviation equation for two specific geodesics, but also we would like to implement for nearby geodesics, a family which is parameterized by $\epsilon$. The differences between the deviations equations for different pair of geodesics is implemented on the initial conditions for the deviation functions $(\chi(0),\chi(0))$.

Therefore, we have the following
\\
\\
{\bf Criteria for consistent $k$-jet field approximation}.
For $k\in \,\mathbb{N}$, the error of a given approximation scheme must be bounded by the error in the $k$-jet approximation of $\Xi^{\mu}(\epsilon,s)$ by $\Omega^\mu_k(\epsilon,s)$, which is of order $k+1$ in $\epsilon$.
\bigskip

The error in considering $\xi^\mu\xi^\nu$ negligible must be comparable to the error of considering $\Xi^\mu_k\Xi^\nu_k$ negligible and  to be compatible with the Taylor's approximation $\Xi^\mu$ by $\Omega^\mu_k$, which implies an error of order $\epsilon^{k+1}$. Therefore, one needs for consistency that the condition
$\Omega^\mu_k\Omega^\nu_k\simeq \mathcal{O}(\epsilon^{k+1})$ holds good.
Similarly, the following conditions must hold for the monomials:
\begin{align}
\Omega^\mu_k\Omega^\nu_k\simeq \mathcal{O}(\epsilon^{k+1}),\quad\dot\Omega^\mu_k\Omega^\nu_k\simeq \mathcal{O}(\epsilon^{k+1}),\quad \dot\Omega^\mu_k\dot\Omega^\nu_k\simeq \mathcal{O}(\epsilon^{k+1}),\quad etc...
\label{criteriaofnegligible}
\end{align}
 The conditions \eqref{criteriaofnegligible} can be checked order by order in $k\in \mathbb{N}$. As a consequence of such analysis, we shall prove the following
 \begin{theorem-A}
Let $\mathcal{S}\in\, \Gamma \,T\mathcal{C}$ be a semi-spray whose associated connection $\nabla$ is such that the Riemann curvature endomorphism $\mathcal{R}$ is non-zero.
  Given two geodesics ${\bf X}:I\to M$ and ${\bf x}:I\to M$ with ${\bf X}(I)\subset V$, ${\bf x}(I)\subset V$, the only $k$-jet  approximation scheme to the exact deviation equation \eqref{standardexactdeviationequation} such that the fields $\{\xi^\mu_k\xi^\nu_k\}^n_{\mu,\nu=1}$ are negligible and compatible with the action of $\Gamma(\mathbb{R}^n)$ is for $k=1$ and corresponds to the standard Jacobi equation.
 \label{Theorem A}
  \end{theorem-A}

  \begin{corollary-B}
  Let $\mathcal{S}\in\,\Gamma\, T\mathcal{C}$ be a spray. If
   \begin{itemize}
\item The solutions of the equation \eqref{eq:genJacobiEq} are $k$-jet fields along the central geodesic ${\bf X}:I\to M$ and

\item The Riemann curvature endomorphism $\mathcal{R}$ is non-zero,
\end{itemize}
then equation \eqref{eq:genJacobiEq} is not compatible with the set of transformations $\Gamma (\mathbb{R}^{n})$.
  \end{corollary-B}

  Given the above Taylor's expansions of order $k$ and a particular approximation scheme, an algebraic expression $G(\xi,\dot{\xi},\ddot{\xi})=0$  {\it can be approximated} by another algebraic expression $D(\xi,\dot{\xi},\ddot{\xi})=0$ by equating to zero in $G(\xi,\dot{\xi},\ddot{\xi})=0$ the monomials that are negligible. Such approximation is consistent with $k$-jet expansions if the error in approximating the expression $G$ by the expression $D$ is bounded or of the same order than the error in the approximation
  \begin{align}
  \{\xi^{\mu}\mapsto \Omega^\mu_k(\epsilon,s),\,\mu=1,...,n\}.
  \label{approximation}
   \end{align}
   If the error $G\to D$ is not bounded by the error in \eqref{approximation}, then the functions $\Omega^\mu_k$, which are the $k$-jet approximation to deviation equations that we are interested, cannot be a reasonable approximation to the solutions of the equation $G(\Omega^\mu_k,\dot{\Omega}^\mu_k, \ddot{\Omega}^\mu_k)=0$. This consistency problem arises clearly when studying the consistence of the approximation schemes under the action of the pseudo-group $\Gamma(\mathbb{R}^n)$.
\end{section}
\section{{Proofs}}
In this {\it section} we prove {\it Theorem} A and {\it Corollary} B. We assume that all the functions are smooth.
\begin{lemma}
The only non-trivial approximation schemes compatible with $\Gamma(\mathbb{R}^n)$ such that the monomials $\{\xi^\mu\xi^\mu,\mu,\nu=1,...,n\}$ are negligible is
 for Taylor expansions with $k=1$ such that the monomials
 \begin{align}
\{\xi^\mu\xi^\nu,\, \xi^\mu\dot{\xi}^\nu,\,\dot{\xi}^\mu\dot{\xi}^\nu,\,\dot{\xi}^\mu\dot{\xi}^\nu\xi^\rho,...\}
\label{negligible monomials}
\end{align}
are negligible.
\label{possibleapproximationschemes}
\end{lemma}
\begin{proof} We start assuming non-trivial approximation schemes. First, we can investigate the case of linear approximation scheme.
For k=1, the monomials \eqref{negligible monomials} can be approximated by the respective monomials
\begin{align*}
\{\Omega^\mu_1\Omega^\nu_1,\, \Omega^\mu_1\dot{\Omega}^\nu_1,\dot{\Omega}^\mu_1\dot{\Omega}^\nu_1,
\dot{\Omega}^\mu_1\dot{\Omega}^\nu_1\Omega^\rho_1,...\}
 \end{align*}
 with an error of order $\epsilon^2$ without imposing any restriction on the character of the functions $\Omega^\mu_1(\epsilon,s)$. Therefore, such approximation scheme is compatible with $\Gamma(\mathbb{R}^n)$ (note that by construction $\epsilon$ can be chosen to be a scalar parameter and therefore, invariant under local coordinate transformations). Indeed, for $k=1$ the monomial $\xi^\mu\xi^\nu$ is negligible, since it is of order $k+1=2$ in $\epsilon$,
\begin{align*}
\Omega^\mu_1(\epsilon,s){\Omega}^\nu_1(\epsilon,s) = {\epsilon}^2\xi^\mu_1(s)\xi^\nu_1(s),\quad \forall\, \epsilon\in (-\epsilon_0,\epsilon_0).
\end{align*}
This implies that the term $\Omega^\mu_1(\epsilon,s){\Omega}^\nu_1(\epsilon,s)$ is of the same order in $\epsilon$ than the error in the approximation $\Xi^\mu(\epsilon, s)\to \Omega^\mu_1(\epsilon,s)$ in any coordinate system, for each $\epsilon\in \, (-\epsilon_0,\epsilon_0)$. Since we have assumed that the remainder term $\rho^\mu_2$ is negligible, one can also neglect the term $\Omega^\mu_1(\epsilon,s){\Omega}^\nu_1(\epsilon,s)$.
Similarly, for the monomial $\Omega^\mu_1(\epsilon,s)\dot{\Omega}^\nu_1(\epsilon,s)$ one has the relations
\begin{align*}
\Omega^\mu_1(\epsilon,s)\dot{\Omega}^\nu_1 (\epsilon,s)& =\,\epsilon\,\xi^\mu_1(s)\,{\epsilon}\,\dot{\xi}^\nu_1(s) = \, \epsilon^2\,\xi^\mu_1(s)\dot{\xi}^\nu_1(s)\,=\mathcal{O}(\epsilon^2).
\end{align*}
This approximation leads to the Jacobi equation for $\xi_1$, which is a covariant equation.

 If we require that for k=1 the monomial  $\xi^\mu\dot{\xi}^\nu$ is not negligible, then from the above argument it follows that $\Omega^\mu_1(\epsilon,s)\dot{\Omega}^\nu_1(\epsilon,s)$ must be of order ${\epsilon}$. Thus, it must exists a smooth tensor $C^{\mu\nu}_1:I\to \mathbb{R}$ such that
\begin{align}
\Omega^\mu_1(\epsilon,s)\dot{\Omega}^\nu_1(\epsilon,s)=C^{\mu\nu}_1(s)\,{\epsilon}, \quad \forall\,\epsilon\in\,(-\epsilon_0,\epsilon_0).
\label{anoncovariant0}
\end{align}
holds, or equivalently
 \begin{align}
\xi^\mu_1(s)\dot{\xi}^\nu_1(s)=C^{\mu\nu}_1(s)\frac{1}{\epsilon}, \quad \forall\,\epsilon\in\,(-\epsilon_0,\epsilon_0).
\label{anoncovariant}
\end{align}
 However, the condition \eqref{anoncovariant} is not compatible with the action of $\Gamma(\mathbb{R}^n)$ except if
\begin{enumerate}
\item If $C^{\mu\nu}_1(s)=0$, implying that each Jacobi field  $\xi^\mu_1\frac{\partial}{\partial x^\mu}$ is zero or $\dot{\xi}^\mu_1\frac{\partial}{\partial x^\mu}$ is zero or both conditions hold.

\item The parameter $\epsilon=1$, a contradiction with our assumptions.
\end{enumerate}

For $k\geq 2$, the argument is analogous. Let us consider the Taylor approximations of order $k$ for $\xi^\mu$ and the corresponding expansion of the monomial
\begin{align*}
 \Omega^\mu_k(\epsilon,s){\Omega}^\nu_k(\epsilon,s) &= \,\Big(\sum^{k}_{j=1}\,\frac{1}{j!}\,\epsilon^j\,\frac{\partial^j \Lambda^\mu(\epsilon,s)}{\partial \epsilon ^j}\Big|_{\epsilon=0}\Big)\Big(\sum^{k}_{j=1}\,\frac{1}{k!}\,\epsilon^j\,\frac{\partial^j \Lambda^\nu(\epsilon,s)}{\partial \epsilon ^j}\Big|_{\epsilon=0}\Big)\\
 = &\,\sum^{k}_{j,l=1}\,\frac{1}{j!l!}\epsilon^{j+l}\,\frac{\partial^j \Lambda^\mu(\epsilon,s)}{\partial \epsilon ^j}\Big|_{\epsilon=0}\,\frac{\partial^l \Lambda^\nu(\epsilon,s)}{\partial \epsilon ^l}\Big|_{\epsilon=0}.
\end{align*}
Each of the monomials must be negligible, which implies that they must be of the same order $\epsilon^{k+1}$ than the remainder $\varrho^\mu_{k+1}(\epsilon,s)$,
\begin{align*}
\epsilon^{j+l}\,\frac{\partial^j \Lambda^\mu(\epsilon,s)}{\partial \epsilon ^j}\Big|_{\epsilon=0}\,\frac{\partial^l \Lambda^\nu(\epsilon,s)}{\partial \epsilon ^l}\Big|_{\epsilon=0}=\,\mathcal{O}(\epsilon^{k+1}),\quad k\geq 2
\end{align*}
For $j=l=1$, this statement implies the condition
\begin{align}
\frac{\partial \Lambda^\mu(\epsilon,s)}{\partial \epsilon }\Big|_{\epsilon=0}\,\frac{\partial\Lambda^\nu(\epsilon,s)}{\partial \epsilon}\Big|_{\epsilon=0}=\,\mathcal{O}(\epsilon^{1})
\label{conditionofsmallerrorforhigherorderk}
\end{align}
 In a similar way as for the case $k=1$, one can prove that the condition \eqref{conditionofsmallerrorforhigherorderk} is not covariant under the action of the set of transformations $\Gamma(\mathbb{R}^n)$. Therefore, non-trivial approximations schemes such that $\xi^\mu\xi^\nu$ are negligible only can work for $k=1$ or when one of the possibilities {\it (i)} or {\it (ii)} above hold.
\end{proof}

\begin{lemma}
If a given approximation scheme
\begin{align*}
(\Xi,\dot{\Xi},\ddot{\Xi}) \to (\Omega_k,\dot{\Omega}_k,\ddot{\Omega}_k)
\end{align*}
is not compatible with the action of the pseudo-group $\Gamma(\mathbb{R}^n)$ and if $G(\chi,\dot{\chi},\ddot{\chi})=0$ is an analytic relation, then the approximation
\begin{align*}
G(\Xi,\dot{\Xi},\ddot{\Xi})=0\,\to\,G(\Omega_k,\dot{\Omega}_k,\ddot{\Omega}_k)=0
\end{align*}
is not consistent with $\Gamma(\mathbb{R}^n)$.
\label{Lemma 4.2}
\end{lemma}
\begin{proof}
If an approximation scheme
$(\Xi,\dot{\Xi},\ddot{\Xi}) \to (\Omega_k,\dot{\Omega}_k,\ddot{\Omega}_k)$ is not compatible with $\Gamma(\mathbb{R}^n)$, neither the  approximation $G(\Xi,\dot{\Xi},\ddot{\Xi})=0\,\to\,G(\Omega_k,\dot{\Omega_k},\ddot{\Omega}_k)=0$ will be, except for a discrete set of values of $\epsilon^*$, solutions of algebraic conditions, as consequence of the finiteness of the degree $k$.
\end{proof}
\begin{lemma}
The expression for the generalized Jacobi equation \eqref{eq:genJacobiEq} is an algebraic relation.
\label{Lemma 4.3}
\end{lemma}
\subsection*{Proof of Theorem~A}
For $k=1$ it follows from {\it Lemma} \ref{possibleapproximationschemes} that the approximation scheme for which all the monomials $\{\xi^\mu\xi^\mu\}$ are negligible is compatible with $\Gamma(\mathbb{R}^n)$,  the rest of the monomials $\xi^\mu\dot{\xi}^\nu$, etc... must also be negligible.

For $k\geq 2$, it follows from {\it Lemma} \ref{possibleapproximationschemes} that the approximation scheme where $\xi^\nu\xi^\mu$ is negligible is not compatible with $\Gamma(\mathbb{R}^n)$ except in the situations when $\epsilon$ is not small, which is a contradiction with the requirement that $\{\xi^\mu \xi^\nu,\,\mu,\nu=1,...,n\}$ are negligible. Then by application of {\it Lemma} \ref{Lemma 4.2} and {\it Lemma}  \ref{Lemma 4.3} the result follows. \hfill$\Box$

{\it Corollary} B follows directly from {\it Theorem} A. If there are no additional restrictions on the curvature endomorphisms (in the form of constraints on the associated Jacobi fields), for $k=1$ the assumptions under which equation \eqref{standardexactdeviationequation} is approximated by \eqref{eq:genJacobiEq} does not hold in arbitrary coordinate systems for jet fields solutions. Therefore, if at least one of the components $(\xi^\mu(s), \dot{\xi}^\mu(s))$ is non-zero, for $k=1$ the only approximation scheme compatible with $\Gamma(\mathbb{R}^n)$ is the linear approximation scheme, leading to the Jacobi equation \eqref{Jacobiequation}, that we know is compatible with the set of transformations $\Gamma(\mathbb{R}^n)$.

\section{{Non-linear approximation schemes and higher order geodesic deviation equations. Application to pseudo-Finsler structures}}

We have until now only considered the possible approximation schemes for the exact geodesic deviation equation such that the monomials $\{\xi^\mu\xi^\nu,\,\mu,\nu=1,...,n\}$ are negligible, showing that for an arbitrary spray, the only approximation scheme consistent with the action of $\Gamma(\mathbb{R}^n)$ corresponds to the Jacobi equation of the connection $\nabla$. However, {\it Theorem A} leaves open the possibility for the existence of alternative consistent approximation schemes if the monomials $\{\xi^\mu\xi^\nu,\mu,\nu=1,...,n\}$ are not negligible. In this case, one can still use Taylor's expansion on the parameter $\epsilon$ and $k$-jet field approximations. This method is a direct generalization to arbitrary sprays of Ba{\.z}a{\'n}ski's theory \cite{Bazanski1977}, originally formulated for Lorentzian  metrics.

Let $\mathcal{S}$ be a spray defined on the sub-manifold $\mathcal{C}\hookrightarrow TM$. Then we can  consider the expansions \eqref{omegakfield}, their time derivatives  and insert them in the exact deviation equation \eqref{standardexactdeviationequation}. The connection coefficients are also expanded in the variable $\epsilon$ around $\epsilon'=0$, obtaining a formal series in $\epsilon$,
\begin{align*}
&\ddot\xi^{\mu}\,+\,\Gamma^{\mu}_{\nu\sigma}(X+\xi,\dot{X}+\dot{\xi})\,
\Big( \dot X^{\nu}+ \dot \xi^{\nu} \Big)
\Big( \dot X^{\sigma}+ \dot\xi^{\sigma}\Big)
 -\,\Gamma^{\mu }_{\sigma\nu}(X)\,
\dot X^{\sigma} \,\dot X^{\nu}=0\quad \Rightarrow \\
&\sum^\infty_{k=1}\,\epsilon^k \,G_k(\Xi^\mu,\dot\Xi^\mu,\ddot{\Xi}^\mu)=0.
\end{align*}
Equating to zero each term of the series, a hierarchy of ordinary differential equations is obtained,
\begin{align}
G_k(\Xi^\mu,\dot\Xi^\mu,\ddot{\Xi}^\mu)=0,\quad k=1,2,3,...
\end{align}
The equation obtained from the first order $G_1(\Xi^\mu,\dot\Xi^\mu,\ddot{\Xi}^\mu)=0$ is the Jacobi equation of the connection $\nabla$ associated with the spray $\mathcal{S}$. Higher order deviation equations are obtained by equating to zero the expressions $G_k(\Xi^\mu,\dot\Xi^\mu,\ddot{\Xi}^\mu)=0$ for $ k=2,3,...$. By construction, these higher order geodesic deviation equations are general covariant, since each of the monomial terms on each expression $G_k=0$  contribute up to the same order in $\epsilon$ in the exact deviation equation, and only terms of order $\epsilon^{k+1}$ are neglected.
\subsection{Higher order geodesic deviations in pseudo-Finsler geometry}
In the generalized Ba{\.z}a{\'n}ski's theory the spray $\mathcal{S}$ is not  necessarily affine, that is, the connection coefficients of $\Gamma^\mu_{\nu\rho}$ are not necessarily functions on $M$. In particular, one can consider {\it Finsler sprays} living on the slit tangent bundle $N=TM\setminus\{0\}$. A direct generalization of the notion of Finsler spacetime  in \cite{Beem:1970} to the analogous notion in pseudo-Finsler structures,
\begin{definition}
A pseudo-Finsler structure of signature $(p,q)$ is a pair $(M,L)$ where
\begin{itemize}
\item $M$ is an $n$-dimensional real and smooth manifold,
\item $L:N\longrightarrow R$ is a real smooth function such that
\begin{itemize}
\item $L(x,\cdot)$ is positive homogeneous of degree two in the variable $\dot{x}$,
\begin{align}
L(x,k\dot{x})=\,k^2\,L(x,\dot{x}),\quad \forall\, k\in ]0,\infty[,
\label{homogeneouscondition}
\end{align}
\item The {\it vertical Hessian}
\begin{align}
g_{\mu\nu}(x,\dot{x})=\,\frac{\partial^2\,L(x,\dot{x})}{\partial \dot{x}^\mu\,\partial \dot{x}^\nu}
\label{nondegeracy-signature}
\end{align}
is non-degenerate and with fixed signature $(-,..^p.,-,+,..^q.,+)$ for all $(x,\dot{x})\in\, N$.
\end{itemize}
\end{itemize}
\label{pseudoFinsler}
\end{definition}
The existence of a basis for the topology of $M$ and topological separability are assumed as usually it is  done in differential geometry (second countable and Hausdorff property).
A pseudo-Finsler structure of Lorentzian signature $(-1,1,...,1)$ is a Finsler spacetime in the sense of J. Beem \cite{Beem:1970}. A pseudo-Finsler structure of positive signature $(+,...,+)$ is a standard Finsler structure \cite{BaoChernShen:2000}. A key point of {\it Definition} \ref{pseudoFinsler} is that, in the Lorentzian signature case, it allows for Finslerian light-like curves, suitable to describe light rays in Finslerian theories of gravity.

The geodesics of $L$ are the critical points of the proper-time functional associated with $L$ (see for instance \cite{GallegoPiccioneVitorio:2012} for the case of Lorentzian signature). Given a pseudo-Finsler structure $(M,L)$,
a spray $\mathcal{S}\in\,\Gamma \,TN$ is the Finsler spray of $L$ if the integral curves of $\mathcal{S}$ are the total lift to $N$ of the geodesics of the pseudo-Finsler function $L$.

Given a pseudo-Finsler structure $(M,L)$, there is a standard
 Ehresmann connection determined by $L$ defining a decomposition
\begin{align}
TN=\mathcal{H}\oplus\mathcal{V},
\end{align}
where $\mathcal{V}=\,ker(d\pi)$. To construct this connection we follow \cite{GallegoPiccioneVitorio:2012}. Let us consider the Cartan's tensor
\begin{align*}
C_{\mu\nu\rho}:=\,\frac{1}{2}\, \frac{\partial g_{\nu\rho}}{\partial \dot{x}^\mu}, \quad \mu,\nu,\rho=1,...,n.
\end{align*}
The formal second kind Christoffel's symbols are defined by the expression
\begin{align*}
\gamma^\mu_{\nu\rho}:= \frac{1}{2}\,g^{\mu\lambda}\,\left(\frac{\partial g_{\nu\lambda}}{\partial x^\rho}+\,\frac{\partial g_{\lambda\rho}}{\partial x^\nu}-\,\frac{\partial g_{\nu\rho}}{\partial x^\lambda}\right),\,\quad \mu,\nu,\rho, \lambda=1,...,n.
\end{align*}
Then the non-linear connection coefficients are defined in $N$ by the expression
\begin{align*}
N^\mu_\nu = \gamma^\mu_{\nu\rho}\,\dot{x}^\rho-\,C^\mu_{\nu\rho}\,\gamma^\rho_{\lambda \sigma}\,\dot{x}^\lambda \,\dot{x}^\sigma,\quad \mu,\nu,\rho,\sigma,\lambda = 1,...,n,
\end{align*}
where $C^\mu_{\nu\rho}=\,g^{\mu\lambda} \, C_{\lambda \rho \sigma}.$

Given the non-linear connection, an adapted frame to the horizontal-vertical decomposition is determined by the smooth tangent basis for ${ T}_u { N}$ for each $\, u\in { N}$:
\begin{align}
\begin{split}
&\left\{ \frac{{\delta}}{{\delta} x^{1}}|_u ,...,\frac{{\delta}}{{\delta} x^{n}} |_u, 
\frac{\partial}{\partial \dot{x}^{1}} |_u,...,\frac{\partial}{\partial
\dot{x}^{n}} |_u\right\},\\
 &\frac{{\delta}}{{\delta} x^{\nu}}|_u =\frac{\partial}{\partial
x^{\nu}}|_u -N^{\mu}\,_{\nu}\frac{\partial}{\partial \dot{x}^{\mu}}|_u ,\quad
\mu,\nu=1,...,n,
\end{split}
\label{adaptedframe}
\end{align}
where $\{N^{\mu}\,_{\nu}(x,\dot{x})\}^{n}_{\mu,\nu=1}$ are the {\it non-linear connection coefficients} associated to the Finsler pseudo-Finsler structure $(M,L)$.
Given a tangent vector $X\in\,T_x M$ and $u\in \pi^{-1}(x)$, there is a unique horizontal tangent
vector $h(X)\in \,T_u N$ with $d\pi(h(X))=X$ (horizontal lift of $X$).

The pull-back bundle $\pi^* {TM}$ is the maximal subset of the Cartesian product
${ N}\times{ TM}$ such that the diagram
\begin{align*}
\xymatrix{\pi^*{TM} \ar[d]_{\pi_1} \ar[r]^{\pi_2} &
{ TM} \ar[d]^{\pi_0}\\
{ N} \ar[r]^{\pi} & { M}}
\end{align*}
commutes. This construction is of relevance for pseudo-Finsler structures, since there is a {\it Chern's type linear connection}  defined on $\pi^* {TM}$, similarly to the connection constructed in the Lorentzian case \cite{GallegoPiccioneVitorio:2012}. If ${\bf X}\in T_{\bf x}M$, its horizontal lift at the point ${\bf u}\in N$ is denoted by ${\bf X}^h_{\bf u}$. The covariant derivative of a {\it section} $\pi^*{\bf Y}=\,Y^\mu\pi^* \frac{\partial}{\partial x^\mu}\in\,\Gamma \,\pi^*{TM}$ at the point ${\bf u}=({\bf x},\dot{\bf x})$ and ${\bf X}\in T_{\bf x} M$ is
\begin{align*}
\nabla_{\bf X} {\bf Y}:=\,\nabla_{{\bf X}^h_{\bf u}}\, \big(Y^\mu\pi^*\frac{\partial}{\partial x^\mu}\big)\big|_{\bf u}=\,{ X}^\nu\frac{\partial Y^\mu}{\partial x^\nu}\frac{\partial}{\partial x^\mu}\big|_{\bf u}+\,\Gamma^\rho_{\mu\nu}(x,\dot{x})\,Y^\mu X^\nu\frac{\partial}{\partial x^\rho}\big|_{\bf u}.
\end{align*}
If parameterized by the proper time parameter of $L$, the geodesics of $L$ correspond to the auto-parallel curves of the Chern's connection,
\begin{align}
\nabla_{\dot{\bf X}}\,\dot{\bf X}=0,
\end{align}
or in local coordinates, the geodesic equation
\begin{align}
\ddot{X}^\mu(s)+\,\Gamma^\mu\,_{\nu\rho}(X,\dot{X})\dot{ X}^\nu\dot{ X}^\rho=0.
\end{align}

For a generic pseudo-Finsler structure $L$ there are only two types of non-trivial curvatures associated with the Chern's type connection, the Riemannian type or $hh$-curvature and the $hv$-vertical curvature, since the vertical or $vv$-curvature is identically zero.
The Riemann type curvature is defined to be the tensor field $R$  along $\pi:N\to M$  with components in local coordinates given by the expression
\begin{align}
R^\mu_{\nu\rho\sigma}(x,\dot{x})=\,\Big(\frac{\delta}{\delta x^\sigma}\Gamma^\mu_{\nu\rho}-\,\frac{\delta}{\delta x^\nu}\Gamma^\mu_{\sigma\rho}+\,\Gamma^\mu_{\rho\lambda}\Gamma^\lambda_{\nu\sigma}-\,\Gamma^\mu_{\sigma\lambda}\Gamma^\lambda_{\rho\nu}\Big)(x,\dot{x}),\quad \mu,\nu,\rho,\sigma=1,...,n.
\label{Riemanntypecurvature}
\end{align}
 The $hv$-{curvature} of the spray $\mathcal{S}$ is the tensor  along $\pi:N\to M$ whose components are
\begin{align}
P^\mu_{\nu \rho\sigma}(x,\dot{x}):=\,\frac{\partial \Gamma^\mu_{\nu\rho}}{\partial \dot{x}^\sigma}(x,\dot{x}),\quad \mu,\nu,\rho,\sigma=1,...,n.
\label{Landsbergtypecurvature}
\end{align}
The tensors $R$ and $P$ are related by Bianchi identities associated
with {\it torsion-free} and {\it almost-metric compatibility}
\cite{BaoChernShen:2000}.
\subsection{Application of Ba{\.z}a{\'n}ski's theory to pseudo-Finsler sprays}
In the following, we present the first and second order geodesic deviation equations for pseudo-Finsler structures following the generalized  Ba{\.z}a{\'n}ski's theory. Let us consider the expansion of the ribbon coordinate functions \eqref{ribboncoordinatefunctions}.
 The covariant vector ${\bf J}_2$ is defined by the expression ${\bf J}_2 =J^\mu_2\frac{\partial}{\partial x^\mu}$, where the components are defined as in \cite{KernerColistetevanHolten:2001},
\begin{align}
J^\mu_2=\,\Xi^\mu_2+\,\Gamma^\mu_{\nu\rho}\Xi^\nu_1\Xi^\rho_1.
\end{align}
The vertical lift of a tangent vector ${\bf Z}\in T_{\bf x}M$ to $T_{({\bf x},{\bf y})}N$ is denoted by ${\bf Z}^v=Z^\mu\frac{\partial}{\partial \dot{x}^\mu}$.
\begin{proposition}
Let $L:N\to \mathbb{R}$ be a Finsler function and consider the expansions given by \eqref{ribboncoordinatefunctions}.
Then
\begin{itemize}
 \item The first order geodesic deviation equation of a Finsler spray is the Jacobi equation for ${\bf \Xi}_1$,
\begin{align}
\nabla_{\dot{\bf X}}\,\nabla_{\dot{\bf X}}\, {\bf \Xi}_1+\,R_{\dot{\bf X}}({\bf \Xi}_1,\dot{\bf  X}){\dot{\bf X}}=\,0
\label{jabociequationfinsler}
\end{align}
with initial conditions $({ \Xi}^\mu_1(0),\dot{\Xi}^\mu_1(0))$.

\item The second order geodesic deviation equation is the following non-linear differential equation for ${\bf J}_2$:
\begin{align}
\begin{split}
& \nabla_{\dot{\bf X}}\,\nabla_{\dot{\bf X}}\, {\bf J}_2+\,R_{\dot{\bf X}}({\bf J}_2,\dot{\bf  X}){\dot{\bf X}} =\,\Big(\nabla_{{\bf \Xi}_1}R\Big)(\dot{\bf X},{\bf \Xi}_1)\dot{\bf X}-\Big(\nabla_{\dot{\bf X}}R\Big)(\dot{\bf X},{\bf \Xi}_1){\bf \Xi}_1\\
& +\,4 R(\dot{\bf X},{\bf \Xi}_1)(\nabla_{\dot{\bf X}}{\bf \Xi}_1)
 -\, \Big(\nabla_{(\nabla_{\dot{\bf X}}{\bf \Xi}_1)^v}P\Big)((\nabla_{\dot{\bf X}}{\bf \Xi}_1)^v,\dot{\bf X})\dot{\bf X}
-2\big(\nabla_{{\bf \Xi}_1}P\Big)((\nabla_{\dot{\bf X}}{\bf \Xi}_1)^v,\dot{\bf X})\dot{\bf X}
\end{split}
\label{secondorderdeviationonfinsler}
\end{align}
with the initial conditions $(J^\mu_2(0),\dot{J}^\mu_2(0))$.
\end{itemize}
\label{propositionondeviationequationinFinsler}
\end{proposition}
\begin{proof}
 Given a pseudo-Finsler structure $(M,L)$ with the corresponding Chern's type connection $\nabla$ and the corresponding connection coefficients $\Gamma^\mu\,_{\nu\rho}(X,\dot{X})$  in a local coordinate chart, the equation \eqref{standardexactdeviationequation} can be expanded in powers of $\epsilon$. Equation \eqref{jabociequationfinsler} is a generalization of the Jacobi equation for pseudo-Finsler structures \cite{GallegoPiccioneVitorio:2012}, obtained by grouping together all the terms proportional to $\epsilon$ in the expansion of the exact geodesic deviation equation \eqref{standardexactdeviationequation}. When this is done, one obtains the equations
\begin{align*}
\ddot \Xi^{\mu}_1\,+\pd{\Gamma^{\mu}_{\nu\sigma}}{x^\rho} (X,\dot{X})\,\Xi^{\rho}_1
\dot X^{\nu}\,\dot X^{\sigma}
+2\,\Gamma^{\mu }_{\nu\sigma}(X,\dot{X})\,\dot X^{\sigma} \Xi^{\nu}_1=0,\quad \,\mu,\nu,\rho,\sigma=1,...,n.
\end{align*}
After a re-arrangement of this expression one obtains equation \eqref{jabociequationfinsler}.

Equation \eqref{secondorderdeviationonfinsler} follows from the equality $G_2=0$ in front of the term $\epsilon^2$ in the exact deviation equation. Re-arranging the terms that are proportional to $\epsilon^2 $, one obtains the following relation
\begin{align*}
F^\mu_T& := \,\ddot \Xi^{\mu}_2\, +\pd{\Gamma^{\mu}_{\nu\sigma}}{x^\rho}(X,\dot{X}) \,\Xi^{\rho}_2
\dot X^{\nu}\,\dot X^{\sigma} +2\,\Gamma^{\mu }_{\nu\sigma}(X,\dot{X})\,\dot X^{\sigma} \Xi^{\nu}_2+\,2\,\Gamma^\mu_{\nu\rho}(X,\dot{X})\dot{\Xi}^\nu_1\dot{\Xi}^\rho_1\\
&\,+4\pd{\Gamma^{\mu}_{\nu\sigma}}{x^\rho}(X,\dot{X})\,\Xi^{\rho}_1
\dot X^{\nu}\,\dot \Xi^{\sigma}_1 +\Xi^\lambda_1\Xi^\sigma_1\frac{\partial^2 \Gamma^\mu_{\nu\rho}(X,\dot{X})}{\partial x^\lambda\partial x^\sigma}\dot{X}^\nu\dot{X}^\rho\\
&+\,\dot{\Xi}^\lambda_1\dot\Xi^\sigma_1\frac{\partial^2 \Gamma^\mu_{\nu\rho}(X,\dot{X})}{\partial \dot{x}^\lambda\partial \dot{x}^\sigma}\dot{X}^\nu\dot{X}^\rho+\,2{\Xi}^\lambda_1\dot\Xi^\sigma_1\frac{\partial^2 \Gamma^\mu_{\nu\rho}(X,\dot{X})}{\partial x^\lambda\partial \dot{x}^\sigma}\dot{X}^\nu\dot{X}^\rho=0.
\end{align*}
In this expression there are two different type of terms. The first and second lines in $F^\mu_T$ correspond to the {\it affine terms}, obtained by derivation respect to the $x$-coordinates the connection coefficients,
\begin{align*}
F^\mu_1& :=\ddot \Xi^{\mu}_2\, +\pd{\Gamma^{\mu}_{\nu\sigma}}{x^\rho}(X,\dot{X}) \,\Xi^{\rho}_2
\dot X^{\nu}\,\dot X^{\sigma}
+2\,\Gamma^{\mu }_{\nu\sigma}(X,\dot{X})\,\dot X^{\sigma} \Xi^{\nu}_2+\,2\,\Gamma^\mu_{\nu\rho}(X,\dot{X})\dot{\Xi}^\nu_1\dot{\Xi}^\rho_1\\
&+4\pd{\Gamma^{\mu}_{\nu\sigma}}{x^\rho}(X,\dot{X})\,\Xi^{\rho}_1
\dot X^{\nu}\,\dot \Xi^{\sigma}_1+\Xi^\lambda_1\Xi^\sigma_1\frac{\partial^2 \Gamma^\mu_{\nu\rho}(X,\dot{X})}{\partial x^\lambda\partial x^\sigma}\dot{X}^\nu\dot{X}^\rho.
\end{align*}
Such terms are the same than in the affine second order deviation equation \cite{Bazanski1977, KernerColistetevanHolten:2001}. Furthermore, along the geodesic ${\bf X}:I\to M$, the Cartan tensor contracted with $\dot{\bf X}$ vanishes, $C(\dot{\bf X},\cdot,\cdot)=0$. This implies that the expression for $F^\mu_1$ is equivalent to the covariant expression
\begin{align*}
F^\mu_1& =\nabla_{\dot{\bf X}}\,\nabla_{\dot{\bf X}}\, {\bf J}_2+\,R_{\dot{\bf X}}({\bf J}_2,\dot{\bf  X}){\dot{\bf X}} \Big(\nabla_{{\bf \Xi}_1}R\Big)(\dot{\bf X},{\bf \Xi}_1)\dot{\bf X}-\Big(\nabla_{\dot{\bf X}}R\Big)(\dot{\bf X},{\bf \Xi}_1){\bf \Xi}_1\\
& \,+\,4 R(\dot{\bf X},{\bf \Xi}_1)(\nabla_{\dot{\bf X}}{\bf \Xi}_1),
\end{align*}
where the covariant derivatives are taken at the point $({\bf X},\dot{\bf X})\in\, N$.

 The third line in the expression for $F^\mu_T$ is related with the $hv$-curvature of the Chern connection and is intrinsically a non-affine contribution. Note that the functions $\{\dot{\Xi}^\mu_1,\mu=1,...,n\}$ do not define the components of a vector field along ${\bf X}:I\to M$. In order to define an associated vector field, one can consider the covariant derivatives ${(\nabla_{\dot{\bf X}}{\bf \Xi}_ 1)^v}$ and $\nabla_{{\bf \Xi}_1}{\bf X}$, both evaluated at $({\bf X},\dot{\bf X})$. Furthermore, note that all the derivatives (and therefore, the corresponding connection coefficients) are taken and considered at the points ${\bf u}(s)=({\bf \Xi}(\epsilon,s),\dot{\bf \Xi}(\epsilon,s))\in T_{\bf \Xi}M\setminus \{0\}$. Thus fixed ${\bf u}$ there is an unique affine connection $\bar{\nabla}$ on $\Gamma((-\epsilon_0,\epsilon_0)\times I)\subset M$ determined by the relation
 \begin{align}
 \bar{\nabla}_{\frac{\partial}{\partial x^\nu}}\,\frac{\partial}{\partial x^\rho}:=\,\Gamma^\mu_{\nu\rho}({\bf \Xi},\dot{\bf \Xi})\,\frac{\partial }{\partial x^\mu},\quad \mu,\nu,\rho=1,...,n.
 \end{align}
 Since the vector $\dot{\bf X}(s)\neq 0$, the Chern connection coefficients $\Gamma^\mu_{\nu\rho}({\bf X},\dot{\bf X})$ are smooth. Moreover, the connection $\bar{\nabla}$ is symmetric. Therefore, there are normal coordinates at each point ${\bf X}(s)$ and such coordinates are smooth. Note that the normal coordinate systems can change along the curve ${\bf X}:I\to M$. Thus, for a fixed point ${\bf X}(s)$, the $hv$-curvature terms can be written in these normal coordinate system as
\begin{align*}
\big(\dot{\Xi}^\lambda_1\dot\Xi^\sigma_1\frac{\partial^2 \Gamma^\mu_{\nu\rho}}{\partial y^\lambda\partial y^\sigma}\dot{X}^\nu\dot{X}^\rho +2\,{\Xi}^\lambda_1\dot\Xi^\sigma_1\frac{\partial^2 \Gamma^\mu_{\nu\rho}}{\partial x^\lambda\partial y^\sigma}\dot{X}^\nu\dot{X}^\rho \big)\frac{\partial}{\partial x^\mu} & =\Big(\nabla_{(\nabla_{\dot{\bf X}}{\bf \Xi}_1)^v}P\Big)((\nabla_{\dot{\bf X}}{\bf \Xi}_1)^v,\dot{\bf X})\dot{\bf X}
 \\& +2\Big(\nabla_{{\bf \Xi}_1}P\Big)((\nabla_{\dot{\bf X}}{\bf \Xi}_1)^v,\dot{\bf X})\dot{\bf X},
\end{align*}
where all the derivatives are taken at the point $u=({\bf X}(s),\dot{\bf X}(s))$. Let us mention that all the $\Gamma$-terms have been put equal to zero on the left hand side of the relation.
Since Ba{\.z}a{\'n}ski's method is general covariant, {\it equation} \eqref{secondorderdeviationonfinsler} holds in any coordinate system.
\end{proof}
\section{Applications}
\subsection{Berwald spaces} Berwald spaces are pseudo-Finsler structures among the closest pseudo-Finsler  manifolds.
\begin{definition} A Berwald space is a pseudo-Finsler structure such that $P=0$.
\end{definition}

 Let us consider a generic pseudo-Finsler structure $(M,F)$. One can be interested to know when such a space has the same un-parameterized geodesics of a pseudo-Riemannian structure $(M,h)$. If this is the case, the corresponding Jacobi equation and higher order geodesic deviation equations must be the same for $F$ and for $h$. It is easy to see that this only can happens if $P=0$. Therefore,
 \begin{proposition}
 A necessary condition for the pseudo-Finsler space has the same un-parameterized geodesics than a pseudo-Riemannian manifold is that $P=0$.
 \end{proposition}

\subsection{Riemann-flat pseudo-Finsler structures}
 Another class of interesting spaces are characterized by the following
\begin{definition}
A {\it Riemann-flat pseudo-Finsler structure} is a pseudo-Finsler structure with $R=0$ and $P\neq  0$.
\end{definition}
For positive definite Finsler metrics, it is difficult to find examples of non-Riemannian and non-Minkowskian spaces (in the Finslerian sense, with both curvature tensors $P$ and $R$ null) with Riemannian curvature tensor ${R}=0$, but with $hv$-curvature non-trivial  for the Chern connection. The fish-tank metric (see \cite{BaoRoblesShen}) has $R=0$ and $P\neq 0$ (for the Chern's connection).

In Riemann-flat spacetimes, the first deviation equation \eqref{jabociequationfinsler} reduces to
\begin{align}
\nabla_{\dot{\bf X}}\,\nabla_{\dot{\bf X}}\, {\bf \Xi}_1=0,
\label{Finslerasymptoticallyflat1}
\end{align}
and the second deviation equation \eqref{secondorderdeviationonfinsler} reduces to
\begin{align}
\nabla_{\dot{\bf X}}\,\nabla_{\dot{\bf X}}\, J_2=-\nabla_{(\nabla_{\dot{X}}{\bf \Xi}_1)^v}P((\nabla_{\dot{\bf X}}{\bf \Xi}_1)^v,\dot{\bf X})\dot{\bf X}\,-2\nabla_{{\bf \Xi}_1}P((\nabla_{\dot{\bf X}}{\bf \Xi}_1)^v,\dot{\bf X})\dot{\bf X}.
\label{Finslerasymptoticallyflat2}
\end{align}

Equations \eqref{Finslerasymptoticallyflat1} and \eqref{Finslerasymptoticallyflat2} have potential applications in {\it Finslerian cosmology} (see for instance \cite{KouretsisStathakopoulosStraviros:2010}). In particular, the second equation involves the local anisotropy tensor  $P$.  Even if the space is flat, the relation \eqref{Finslerasymptoticallyflat2} implies a non-trivial behaviour for neighboring geodesics, only observable at large scales.

Other source of applications of these space are  on phenomenology of Finslerian structures on gravitational waves. In this case, due to the long separation of the test particles moving on a wave, higher order deviation equations could be very useful.

\section{{Discussion}}

We have shown that the covariance for the generalized Jacobi equation \eqref{eq:genJacobiEq} is in contradiction with
{\it Lemma} \eqref{possibleapproximationschemes}, since the hypothesis that second order monomials $\xi^\mu\xi^\nu$ etc... are negligible is consistent with general covariance only for $k=1$ and for the linear approximation scheme. Indeed, for rapid deviation schemes, it  is not longer true that the square of the deviation functions are negligible. Thus, some constraints must be imposed to obtain a covariant equivalent version of the equation \eqref{eq:genJacobiEq}.

Equation \eqref{eq:genJacobiEq} is covariant under affine local coordinate transformations. Since the transformation between two Fermi coordinate systems are affine coordinate transformations, the generalized Jacobi equation \eqref{eq:genJacobiEq} is covariant under coordinate transformations from Fermi to Fermi local coordinate systems \cite{DahlTorrome}. However, as an approximation to the exact deviation equation \eqref{standardexactdeviationequation}, the generalized Jacobi equation \eqref{eq:genJacobiEq} fails to be general covariant, since the hypothesis of the rapid approximation scheme break down in arbitrary coordinates.

An alternative theory  to the one based upon the equation \cite{Hodgkinson} was initially developed by B. Mashhoon \cite{Mashhoon1, Mashhoon:1977} and further applied in astrophysical systems (see for instance \cite{ChiconeMashhoon:2002, ChiconeMashhoon:2005, ChiconeMashhoon:2006}).
In Mashhoon's construction, one first considers a local Fermi coordinate system $(\,^Fx,U_F)$ where (at least locally) the image of the two geodesics ${\bf x},{\bf X}:I\to M$ are defined on the domain $U_F\subset M$. When this is possible and such Fermi coordinate exists, the solutions $\{\,^F\xi^{\mu}_M\}^n_{\mu=1}$ of Mashhoon's generalized geodesic deviation  in Fermi coordinates correspond to a tangent vector in the direction of the geodesic joining the central geodesic ${\bf X}(s)$ with the corresponding point of the second geodesic ${\bf x}(s)$. Although initially formulated in Fermi coordinates, Mashhoon's equation can be written in arbitrary coordinates (see {\it Appendix C} in \cite{ChiconeMashhoon:2002} or reference \cite{ChiconeMashhoon:2005} for details). However, the fact that initially we are able to embed the geodesics in a Fermi coordinate chart imposes a constraint, restricting the theory to such special geometric case. Furthermore, the physical meaning attached to the solutions of Mashhoon's equation as describing the deviation functions $^F\xi^{\mu}_M=\xi^\mu$ is only valid for Fermi coordinate systems: as we know (see for instance \cite{DahlTorrome}), $\{\xi^\mu\}^n_{\mu=1}$ are not tensorial when identified with the solutions of equation\eqref{eq:genJacobiEq}, while the components $\{\,^F\xi^{\mu}_M\}^n_{\mu=1}$ in Mashhoon's theory define a vector.

It was discussed by B. Schutz  \cite{Schutz1985} that there is no a consistent generalization of the geodesic deviation equation in the {\it rapidly deviation scheme}. Schutz's analysis relies on the prescription that the geodesic curves are joined by geodesics and the argument is restricted to Riemann normal coordinates. In contrast, we have only required the existence of an initial simple curve connecting the initial points ${\bf x}(0)$ and ${\bf X}(0)$. Thus our method extends the conclusion of \cite{Schutz1985} to more general pairs of geodesics, not necessarily with image in the interior of normal coordinate domains.

We have seen that there are alternative frameworks for generalizing geodesic deviation equations beyond linearization in the geodesic deviation functions  $\{\xi^\mu_j\}^{n,k}_{\mu=1,j=1}$.  Ba{\.z}a{\'n}ski's theory \cite{Bazanski1977} is a convenient framework to investigate geodesic deviations beyond the Jacobi equation. In such formalism, one can formulate an hierarchy of general covariant differential equations for the functions $\{\xi^\mu_j\}^{n,k}_{\mu=1,j=1}$, although one needs to abandon the requirement that the monomials  $\{\xi^\mu_j\xi^\nu_i\}^{n,k}_{\mu,\nu=1,j,i=1}$ are negligible.
Ba{\.z}a{\'n}ski's theory was applied extensively in the investigation of geodesic motion in general relativistic spacetimes (see for instance  \cite{KernerColistetevanHolten:2001, ColisteteLeygnacKerner2002, KoekoekHolten2011} and in subsequent works by these authors).
We have shown that Ba{\.z}a{\'n}ski's theory can be extended to connections determined by arbitrary sprays $\mathcal{S}\in \,\Gamma \,T\mathcal{C}$.
This generalized framework has been applied in this paper to an arbitrary Finsler spray, obtaining the usual Jacobi equation \eqref{jabociequationfinsler} and a new second order deviation equation \eqref{secondorderdeviationonfinsler} in pseudo-Finsler geometry. Some of the applications have been briefly mentioned, as the case of Berwald-type spacetimes and spaces with $P\neq 0$ and $R=0$.

Finally, let us remark that in the literature there are geometric generalizations of the Jacobi equation associated to a general semi-spray. For instance, in the work of I. Bucataru and M. F. Dahl it is considered how $k$-parameter geodesic variations are related with $k$-lift of the semi-spray in the $k$-iterated tangent bundle \cite{BucataruDahl2011}. Such a result applies to general semi-sprays and implies differential conditions for the jets of the semi-spray. Moreover, the theory was also developed for the particular case of sprays, which a generalization of the Jacobi field equation in terms of Jacobi like equation in terms of the Jacobi tensor of the spray \cite{BucataruDahl2011}. How Bucataru-Dahl theory is related with the proposals discussed here, specifically, with Ba{\.z}a{\'n}ski's theory, is an interesting question, that we post-pone for future research.

\subsection*{Acknowledgements} R. G. T was financially supported by FAPESP (Brazil), process 2010/11934-6 and by PNPD-CAPES nº. 2265/2011 (Brazil); J. G. is grateful
for the support provided by STFC (the Cockcroft Institute ST/G008248/1 and ST/P002056/1) and
EPSRC (the Alpha-X project EP/J018171/1 and EP/N028694/1).

\small{}
\end{document}